\documentclass[runningheads]{llncs}
\usepackage{cite}
\usepackage{caption2}
\usepackage{graphicx}
\usepackage{amsmath}
\usepackage{amsfonts}
\usepackage{url}
\begin{document}

\title{Secure Linear Aggregation Using Decentralized Threshold Additive Homomorphic Encryption For Federated Learning}
\author{Haibo~Tian\inst{1} \and
        Fangguo~Zhang\inst{1} \and
        Yunfeng~Shao\inst{2} \and
        Bingshuai~Li\inst{2}
        % <-this % stops a space
\thanks{Dr. Tian and Prof. Zhang were with the GuangDong Province Key Laboratory of Information Security Technology, School of Data and Computer Science,
Sun Yat-Sen University, Guangzhou, Guangdong, 510275, P. R. China e-mail: \{tianhb, isszhfg\}@mail.sysu.edu.cn.
Dr. Shao and Dr. Li were with Huawei Noah’s Ark Lab. e-mail: \{shaoyunfeng, libingshuai\}@huawei.com.
}
% <-this % stops a space
}
\institute{Sun Yat-Sen University\and 
	Huawei Noah’s Ark Lab
}
%\thanks{Manuscript received XX, XXXX; revised XX, XXXX.}}

\markboth{Journal of IEEE Transactions on Information Forensics and Security}%
{Tian \MakeLowercase{\textit{et al.}}: Secure Linear Aggregation using DTAHE for FL}
\maketitle
\begin{abstract}
Secure linear aggregation is to linearly aggregate private inputs of different users with privacy protection. The server in a federated learning (FL) environment can fulfill any linear computation on private inputs of users through the secure linear aggregation. At present, based on pseudo-random number generator and one-time padding technique, one can efficiently compute the sum of user inputs in FL, but linear calculations of user inputs are not well supported. Based on decentralized threshold additive homomorphic encryption (DTAHE) schemes, this paper provides a secure linear aggregation protocol, which allows the server to multiply the user inputs by any coefficients and to sum them together, so that the server can build a full connected layer or a convolution layer on top of user inputs. The protocol adopts the framework of Bonawitz et al. to provide fault tolerance for user dropping out, and exploits a blockchain smart contract to encourage the server honest. The paper gives a security model, security proofs and a concrete lattice based DTAHE scheme for the protocol.  It evaluates the communication and computation costs of known DTAHE construction methods. The evaluation shows that an elliptic curve based DTAHE is friendly to users and the lattice based version leads to a light computation on the server.
\end{abstract}

\begin{keywords}
Privacy Protection, Secure Linear Aggregation, Additive Homomorphic Encryption, Smart Contract.
\end{keywords}

% \fi
%\IEEEpeerreviewmaketitle

\section{Introduction}
Federated learning (FL) is intended to train better machine learning models on decentralized real-world data. The models then could be used to build more intelligent equipments for people, such as cars, wearable devices or browsers. McMahan et al. \cite{MM16} proposed a well known FL protocol. The players in their protocol include users who owned data and a parameter server that aggregates model information of users. The protocol runs periodically. In each period, the parameter server randomly selects some users to upload their local model parameters, and averages the parameters to update a global learning model. A user in a period downloads the global learning model, feeds their local data, runs a deep learning network locally and gets updated local model parameters, the information of which is sent to the parameter server. In different periods the server may select different users, and within a period some of the selected users may drop out.

As real-world data are usually sensitive, an important problem in FL is data privacy. Although user data are not directly sent to the parameter server in FL, information of local model parameters may leak the raw data of users. Fredrikson et al. \cite{FJR15} show how to recover train samples from prediction results of a model. Rubaie and Chang \cite{RA16} exploit feature vectors to reconstruct raw input data of a model. Chai et al. \cite{CWCY19} show how to recover preferences of users by model gradient data. Bonawitz et al. \cite{B17} believe that recent updated local model parameters of a user may leak raw data of the user.

There are mainly two approaches to solve the data privacy problem in FL. One uses differential privacy and the other uses cryptographic tools. The work of Martin et al. \cite{DLDP16} is an earlier report of the differential privacy approach. Wei et al. \cite{WLD20} point out a tradeoff between the convergence performance and privacy protection levels of the differential privacy method. For a fixed privacy protection level, the number of users increases, the convergence performance behaves better. However, if the number of users in each period is limited, one may use the cryptographic tools based approach. Bonawitz et al. \cite{B17} propose an elegant solution based on one-time padding and a secure pseudorandom generator.

To the best of our knowledge, the solution of Bonawitz et al. \cite{B17} is the only work suitable for the FL using cryptographic tools. They take the data privacy problem in FL as an secure aggregation problem. And they show some new requirements of a secure aggregation protocol for the FL. Except a security requirement, other requirements are as follow:
\begin{enumerate}
\item The protocol should operates on a high-dimensional vectors;
\item The protocol should tolerate users dropping out;
\item The protocol should be communication efficient even with a new set of uses on each period.
\end{enumerate}
With these requirements, Bonawitz et al. \cite{B17} show that previous works are unsatisfactory which include some works based on homomorphic encryption schemes. In detail, they believe that solutions based on Paillier scheme \cite{RN10,JK12,JL13,CANS2014} are either computationally expensive or require additional trusted dealer, and solutions based on ElGamal scheme \cite{NDSS11,FC12,PUDA12,LC13} need a high expansion factor considering the size of the group elements and that of the model parameters.

Considering the development of the communication technology, we believe that a moderate expansion factor is acceptable and the functionality of a protocol is more important. Liu et al. \cite{L4MZ20} proposed a federated forest where a parameter server should find the maximal value of user inputs. Zhuo et al. \cite{zhuo2020} proposed a federated deep reinforcement learning model where a parameter server needs to build a multi-layer perception on user inputs. The users in these scenarios are usually not mobile users so that the communication cost is not the dominate factor. However, to protect the privacy of users, a sum only secure aggregation is not enough.

We provide a secure linear aggregation protocol to enrich the parameter server. It could naturally be used in the federated averaging algorithm \cite{MM16} in the same way as the secure sum aggregation \cite{B17}. It also could be used in \cite{zhuo2020} to build a linear multi-layer perception to get a reinforcement learning model. It may be adapted in \cite{L4MZ20} with the homomorphic encryption \cite{real} to find the maximal value of user inputs. For simplicity, the solution here is only based on decentralized threshold additive homomorphic encryption (DTAHE) schemes.

Currently, there are three known methods to construct a DTAHE scheme. The first method is to distribute many secret shares to a user. Bendlin and Damg{\aa}rd \cite{BD10} proposed a threshold homomorphic encryption scheme in this way. It relies on secret sharing schemes with general access structure. Boneh et al. \cite{BGGJK} proposed such a scheme based on secret sharing schemes constructed by a monotone formulae. The second method is to use a large modulus for coefficients of an element in a polynomial ring. Boneh et al. \cite{BGGJK} propose to use Shamir secret sharing scheme in this way. The last method is to use the ElGamal encryption as a basic building tool \cite{P91}. We exclude the Paillier encryption based construction method since it is hard to produce a shared key pair in a distributed manner without a trusted dealer. According to our evaluation, in the secure linear aggregation protocol, the communication overhead of the first method is too high and the computation overhead of the third method is a bit high. The second method also has a drawback since a large modulus means a higher polynomial degree when the noise bound is fixed \cite{FV12}. So we provide a new method to construct a DTAHE scheme with security proofs. It does not increase the modulus size and polynomial degree.

The secure linear aggregation protocol is against an active adversary \cite{B17}. An active adversary could corrupt a parameter server and ask $t$ users to decrypt a cipher of a target user. It is not easy to defend against the attack. Note that Bonawitz et al. \cite{B17} add a consistency check round to solve a similar problem in their protocol. However, even we add such a round, the problem still exists since the target user may have dropped out before the consistency check round. To solve the problem, we introduce a blockchain system. We design a smart contract to record ciphers of users and to check the evaluation process of the parameter server. If the parameter server deviates its expected behaviours in a period, the parameter server could be punished and the users in that period could get compensation.

There exists a lot of works to introduce a blockchain into the FL. A top complained problem in FL is the incentive of users to participate a learning process. Researchers propose blockchain enabled models \cite{ZLY19,BSXHH19,LHZMZ20,QGLXYLZ20,WWZLZL20} to give rewards to data owners. Basically, model initializer proposes model parameters and rewards in the blockchain, data owners choose their interested model to download the current model parameters, to train an updated model by their local private inputs and to update information of model parameters to the blockchain, and miners of the blockchain aggregate inputs of users to get a new global model parameter and rewards relevant users. Pokhrel and Choi \cite{PC20} and Kim et al. \cite{KPBK20} show some theoretical results about the performance of blockchain enabled FL considering the delays in a blockchain. A simulation result in \cite{WGWJYL20} shows that the FL without a blockchain is most efficient. Another motivation to introduce a blockchain system to the FL is about the trustiness of a learning model. Sarpatwar et al. \cite{SVMSHG19} propose to model and capture provenance of an overall learning process for verification. Awan et al. \cite{ALLL19} propose to record data produced in a learning task in a blockchain for verification. We use a blockchain as a trusted third party to verify the behaviours of the parameter server. Since a verification process usually could be separated from a learning process, the FL and the blockchain could work at their own paces.

In summary, our contributions are as follows.
\begin{itemize}
\item We give a definition of DTAHE for FL and provide a basic linear aggregation protocol based on the definition. The basic protocol supports a parameter server to linearly operate ciphers of model parameters from different users.
\item We provide a secure linear aggregation protocol against an active adversary with security model and proofs. It shows how to use a blockchain smart contract to help a secure linear aggregation protocol.
\item We provides a new method to construct a lattice based instance of the DTAHE scheme with security proofs. Evaluations show that the DTAHE scheme leads to a lightweight computation on the server side.
\end{itemize}

\section{Preliminaries}
\subsection{Basic Notations}
For any set $X$, we denote by $|X|$ the number of elements of the set $X$. If $x$ is a string, $|x|$ denotes its bit length. And if $x$ is a vector, $|x|$ denotes the dimension of the vector. $x||y$ denotes the bit catenation of two strings $x$ and $y$.

Let $R = \mathbb{Z}[x]/(f(x))$ be a polynomial ring where $f(x)$ is a monic irreducible polynomial of degree $d$. Elements of the ring $R$ is denoted by vectors. For $\vec{a} \in R$, the coefficients of $\vec{a}$ is denoted by $a_i$ such that $\vec{a} = \sum_{i=0}^{d-1}a_i\cdot x^i$. The infinity norm of $||\vec{a}||$ is defined as $max_i|a_i|$ and the expansion factor of R is defined as $\delta_R = max\{||\vec{a}\cdot \vec{b}||/(||\vec{a}||\cdot ||\vec{b}||):\vec{a},\vec{b} \in R\}$.

Let $h > 1$ be an integer. Then $\mathbb{Z}_h$ denotes a set of integers $(-\frac{h}{2},\frac{h}{2}]$. The symbol $\mathbb{Z}/q\mathbb{Z}$ denotes a ring on integers $\{0,\ldots,q-1\}$. For $x \in \mathbb{Z}$, $[x]_h$ denotes the unique integer in $\mathbb{Z}_h$ with $[x]_h = x \bmod h$. For $\vec{x} \in R$, $[\vec{x}]_h$ denotes the element in $R$ obtained by applying $[\cdot]_h$ to all its coefficients. For $x \in \mathbb{R}$, $\lfloor x \rceil$ denotes rounding to the nearest integer and $\lfloor x \rfloor$, $\lceil x \rceil$ denote rounding up or down.

Let $\lambda$ be an integer as the security parameter. A function $negl(\lambda)$ is negligible in $\lambda$ if $negl(\lambda) = o(1/\lambda^c)$ for every $c \in \mathbb{N}$. An event occurs with negligible probability if the probability of the event is $negl(\lambda)$. An event occurs with overwhelming probability if its complement occurs with negligible probability.

Given a probability distribution $\mathcal{D}$, we use $x \leftarrow \mathcal{D}$ to denote that $x$ is sampled from $\mathcal{D}$. For a set $X$, $x \leftarrow X$ denotes that $x$ is sampled uniformly from $X$. A distribution $\chi$ over integers is called $B$-bounded if it is supported on $[-B,B]$.
\subsection{Federated Learning}
McMahan et al. \cite{MM16} described the FL framework. We give a briefly review to see their federated averaging algorithm. As shown in Fig. \ref{fw1}, there is a parameter server $S_1$ and many users. $S_1$ exchanges model parameters with users to collaboratively build a better learning model. The federated averaging algorithm runs periodically. In each period, $S_1$ selects $n$ users randomly. In Fig. \ref{fw1}, we intended to show two periods with totally different users. Next we focus on a period $e$ where users in a set $U_e$.
\begin{figure}[htbp]
  \centering
  \includegraphics[width=0.4\textwidth]{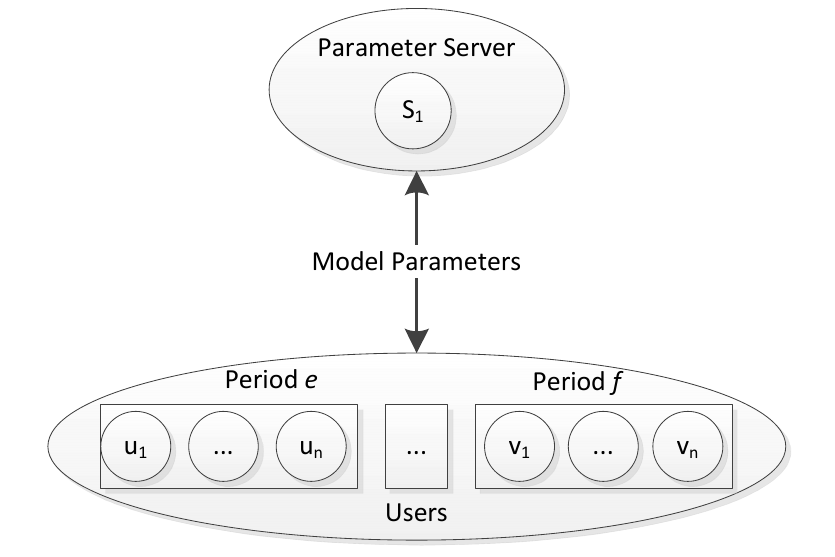}\\
  \caption{The framework of FL}\label{fw1}
\end{figure}

The FL in a period is as follows.
\begin{enumerate}
\item A user $u \in U_e$ is activated to download the global parameters from the parameter server.
\item The user $u$ feeds their local data samples to get updated model parameters. And the parameters or their gradients $\omega_u$ with the number of local data samples $n_u$ should be sent to $S_1$.
\item The parameter server $S_1$ compute the updated global model parameters $\omega = \frac{1}{\sum_{u \in U_e} n_u} \sum_{u \in U_e} n_u\omega_u$.
\end{enumerate}
The parameter server $S_1$ continues to randomly activate a set of users in the next period until the parameters $\omega$ converge to that of the global optimal learning model.

Users could upload the $n_u$ together with $\omega_u$ so that the server could get the value $\sum_{u \in U_e} n_u$ and compute the updated $\omega$. So a sum only aggregation is enough for the federated averaging algorithm. However, a parameter server may execute more computations in other federated learning models \cite{L4MZ20,zhuo2020}, where a sum only aggregation is not enough.

\subsection{The Model Of Secure Aggregation For FL}
\label{promod}
We adapt the notion of message-driven entities in \cite{RG05} to the FL security model \cite{B17} to describe the participants in a FL scenario. A message-driven entity is initially invoked by an environment process. Each entity has their initial states. Once invoked, an entity waits for an activation that can happen for a message from the network or an environment process. On activation, the entity processes the incoming message or the arguments from a process with its current internal state, and generates a new internal state, outgoing messages and a cumulative output. Once an activation is completed, the entity waits for the next activation if it does not stop.

A protocol $\pi$ consists of a server side progress $\pi_S$ and a client side progress $\pi_u$. An FL process in $S_1$ invokes $\pi_S$, which runs periodically until $\pi_S$ is stopped by the FL process. In each period, $\pi_S$ randomly selects a set of users $U$ and activates their client side progresses. A user $u$ should have voluntarily invoked their $\pi_u$ before $S_1$ activates them. On activation, $\pi_u$ exchanges messages with $\pi_S$ until $\pi_S$ finishes a period or $\pi_u$ is stopped by the user. There may be some auxiliary entities for the security of the protocol $\pi$. A common setup includes a certificate authority (CA) \cite{B17}. A $CA$ is invoked and runs permanently to receive certificate requests from users and issue certificates to users. In our protocol, we also introduce blockchain miners $BM$. Miners are invoked to receive transactions from users and the server, and to execute smart contracts on commonly consented transactions.

The server and each user have a bi-directional secure channel between them. Practically, it reuses a transport layer secure channel established by FL processes. Users do not have direct links in the protocol $\pi$. The communications of users are transferred by the server. Users and the $CA$ may exchange messages at an initial stage. When the protocol $\pi$ runs, the $CA$ could be offline. Blockchain miners usually have a point-to-point (P2P) network to deliver transactions and blocks. When a blockchain is used in $\pi$, FL entities should have the ability to send and receive transactions to and from the P2P network.

An active adversary $\mathcal{A}$ could corrupt the server and users in $\pi$. For each period, the number of corrupted parties is at most $n_c$ \cite{B17}. When $\mathcal{A}$ corrupts an entity, it knows the internal states and long-term secrets of the entity, and controls all the behaviours of the entity from that point. Obviously, the server $S_1$ could be corrupted to act as a malicious server. When $\mathcal{A}$ corrupts a user, it could send transactions on behalf of the user.

We does not allow an adversary to corrupt $CA$ or $BM$ since we do not try to model the security of a CA system or a blockchain system. We simply make assumptions about them:
\begin{itemize}
\item A transaction from a user or the server will be executed and recorded correctly by blockchain miners within a bounded delay.
\item A certificate from a $CA$ proves the binding relation of a verification key and the real identity of a user.
\end{itemize}

We define $GO_{\pi, e, \{S_1\}\cup U_e, \mathcal{A}}(\vec{x},\vec{r})$ as a global output for $\pi$ in a period $e$ where $\vec{x}=\{x_S\} \cup \{x_u\}_{u \in U_e}$ denotes the inputs of the server $S_1$ and the inputs of users in the set $U_e$ of the period $e$, and $\vec{r} = \{r_0, r_S\} \cup \{r_u\}_{u \in U_e}$ denotes the random inputs of the adversary $\mathcal{A}$, server $S_1$ and users in $U_e$. Let \[GO_{\pi,e, \{S_1\}\cup U_e, \mathcal{A}}(\vec{x})\] be the random variable about $GO_{\pi,e, \{S_1\}\cup U_e, \mathcal{A}}(\vec{x},\vec{r})$. The randomness comes from the adversary $\mathcal{A}$, server $S_1$ and users in $U_e$.

We define the input privacy of a user \cite{B17} as follows.
\begin{definition}
\label{pri}
(Input Privacy of a User).
For any uncorrupted user $u \in U_e$ in a period $e$, the privacy of its input is assured if \[GO_{\pi,e, \{S_1\}\cup U_e, \mathcal{A}}(\vec{x}) \approx GO_{\pi,e, \{S_1\}\cup U_e, \mathcal{A}}(\vec{x}^\prime)\] where $\vec{x}^\prime = \{x_S\}\cup \{x_v\}_{v \in U_e \backslash \{u\}}\cup\{0\}$ and ``$\approx$'' means computationally indistinguishable.
\end{definition}

The security goal in \cite{B17} is input privacy of all uncorrupted users. Apparently, if the privacy of all users are protected, the privacy of a user is protected. If the privacy of each uncorrupted user in protected, the privacy of all uncorrupted users are protected. So the two definitions are equivalent.

\subsection{DTAHE}
Boneh et al. \cite{BGGJK} give a definition of decentralized threshold fully homomorphic encryption (DTFHE). We refine it as a DTAHE definition for the FL scenario.
A DTAHE scheme is a tuple of probabilistic polynomial time (PPT) algorithms $DTAHE = (Setup$, $KeyGen$, $Share$, $CombKey$, $Enc$, $Eval$, $ParDec$, $FinDec)$.
\begin{enumerate}
\item $Setup(1^\lambda) \rightarrow parm$: It takes as input a security parameter $\lambda$, outputs system parameters $parm$.
\item $KeyGen(parm) \rightarrow (sk_u, pk_u)$: It takes as input the system parameter $parm$ to produce public and private keys for a user $u$.
\item $Share(parm, \{pk_v\}_{v \in U\backslash\{u\}}, t, sk_u) \rightarrow \{e_{v,u}\}_{v \in U\backslash\{u\}}$: It takes as input the system parameter $parm$, public keys of users in a set $U$ excluding the user $u$, a threshold value $t$ and private keys $sk_u$ of the user $u$, to produce encrypted shares $e_{v,u}$ for each user $v \in U \backslash\{u\}$.
\item $CombKey(parm, \{pk_u\}_{u \in U}) \rightarrow pk$: It takes as input the system parameter $parm$, public keys of a set of users in $U$, and produces an encryption key $pk$.
\item $Enc(parm, pk, m_u) \rightarrow c_u$: It takes as input the system parameter $parm$, a public key $pk$ and a message $m_u$ from a user $u$, and produces a ciphertext $c_u$.
\item $Eval(parm, \{c_u\}_{u \in U}, \{\alpha_u\}_{u \in U}) \rightarrow \hat{c}$: It takes as the system parameter $parm$, ciphers $\{c_u\}_{u \in U}$ and coefficients $\{\alpha_u\}_{u \in U}$, and produces an evaluated cipher $\hat{c}=\sum_{u \in U}\alpha_u \cdot c_u$.
\item $ParDec(parm,\hat{c},\{e_{u,v}\}_{v \in U}) \rightarrow \hat{m}_u$: It takes as input the system parameter $parm$, the cipher $\hat{c}$, and a set of encrypted shares $\{e_{u,v}\}_{v \in U \backslash \{u\}}$ to the user $u$, and produces a partially decrypted value $\hat{m}_u$.
\item $FinDec(parm,t,\hat{c},\{\hat{m}_u\}_{u \in V}) \rightarrow m$: It takes as input the system parameter $parm$, the threshold value $t$, the cipher $\hat{c}$ and partially decrypted ciphers $\{\hat{m}_u\}_{u \in V}$ from users in a set $V$ with $|V| \geq t$, and produces a plaintext $m$.
\end{enumerate}

One could simply give an ElGamal based DTAHE instance following the constructions in \cite{P91}, which justifies the correctness of the DTAHE definition.

\subsection{DTAHE Model}
\label{DTAHEModel}
We adapt the model of DTFHE in \cite{BGGJK} for the DTAHE. The first definition is evaluation correctness.
\begin{definition}
\label{corrdef}
(Evaluation Correctness). A DTAHE scheme for a set of users $U$ satisfies evaluation correctness if for all $\lambda$ and $t$, the following holds:

For an evaluated cipher \[\hat{c} \leftarrow Eval( parm, \{c_u\}_{u \in U}, \{\alpha_u\}_{u \in U})\] the probability
\[Pr\left[ \begin{gathered}
  FinDec\left( \begin{gathered}
  parm,t, \hat c, \hfill \\
  {\{ ParDec(parm,\hat c,{\{ {e_{u,v}}\} _{v \in U}})\} _{u \in V}} \hfill \\
\end{gathered}  \right) \hfill \\
   = \sum\limits_{u \in U} {{\alpha _u}}  \cdot {m_u} \hfill \\
\end{gathered}  \right]\]
is overwhelming where \[c_u \leftarrow Enc(parm, pk, m_u),\] \[(\{e_{v,u}\}_{v \in U}) \leftarrow Share(parm, \{pk_v\}_{v \in U\backslash\{u\}}, t, sk_u),\] \[(sk_u, pk_u) \leftarrow KeyGen(parm),\] \[pk \leftarrow CombKey(parm, \{pk_u\}_{u \in U}),\] and \[parm \leftarrow Setup(1^\lambda).\]
\end{definition}

The second definition is sematic security. It captures the privacy of messages.
\begin{definition}
\label{semdef}
(Sematic Security). We say that a DTAHE scheme for a user set $U$ satisfies sematic security if for all $\lambda$, the following holds:

For any PPT adversary $\mathcal{A}$, the following experiments
$Expt_{\mathcal{A},Sem}(1^\lambda)$ outputs $1$ with probability $\frac{1}{2}+negl(\lambda)$:
\begin{itemize}
\item $Expt_{\mathcal{A},Sem}(1^\lambda)$:
    \begin{enumerate}
    \item The adversary outputs $U$ and $V$ where $|U| =n$ and $|V| = t$ specify an access structure.
    \item The challenger runs \[parm \leftarrow Setup(1^\lambda),\] \[(sk_u,pk_u)\leftarrow KeyGen(parm),\]
    \[{\{ {e_{v,u}}\} _{v \in U  \backslash \{u\}}} \leftarrow Share\left( \begin{gathered}
  parm,{\{ p{k_v}\} _{v \in U\backslash \{ u\} }}, \hfill \\
  t,s{k_u} \hfill \\
\end{gathered}  \right),\]
     \[pk \leftarrow CombKey(parm,\{pk_u\}_{u \in U}),\] and provides $(parm, pk, \{\{e_{v,u}\}_{v \in U  \backslash \{u\}}\}_{u \in U})$ to $\mathcal{A}$.
    \item $\mathcal{A}$ outputs a set $S \subseteq U$ such that $|S| < t$. It submits message vectors $\{m_{u,0},m_{u,1}\}_{u \in U}$ and $S$ to the challenger.
    \item The challenger provides $\mathcal{A}$ the shares $\{\{s_{u,v}\}_{v \in U }\}_{u \in S}$ and a cipher set \[\{c_u \leftarrow Enc(parm, pk, m_{u,b})\}_{u \in U}\text{, for }b \in\{0,1\}.\]
    \item $\mathcal{A}$  outputs a  guess bit $b^\prime$. The experiment outputs $1$ if $b = b^\prime$.
    \end{enumerate}
\end{itemize}
\end{definition}

The last definition is simulation security. It captures the privacy of shared secrets and private keys of users.
\begin{definition}
\label{simdef}
(Simulation Security). A DTAHE scheme satisfies simulation security if for all $\lambda$, the following holds:

There is a stateful PPT algorithm $\mathcal{C} = (\mathcal{C}_1, \mathcal{C}_2)$ such that for any PPT adversary $\mathcal{A}$, the following experiments
$Expt_{\mathcal{A},Real}(1^\lambda)$ and $Expt_{\mathcal{A},Ideal}(1^\lambda)$ are indistinguishable:
\begin{itemize}
\item $Expt_{\mathcal{A},Real}(1^\lambda)$:
    \begin{enumerate}
    \item The adversary outputs $U$ and $V$ where $|U| =n$ and $|V| = t$ specify an access structure.
    \item The challenger runs \[parm \leftarrow Setup(1^\lambda),\] \[(sk_u,pk_u)\leftarrow KeyGen(parm),\]
    \[({\{ {e_{v,u}}\} _{v \in U}}) \leftarrow Share\left( \begin{gathered}
  parm,{\{ p{k_v}\} _{v \in U\backslash \{ u\} }}, \hfill \\
  t,s{k_u} \hfill \\
\end{gathered}  \right),\]
    \[pk \leftarrow CombKey(parm,\{pk_u\}_{u \in U}),\] and provides $(parm, pk, \{\{e_{v,u}\}_{v \in U}\}_{u \in U})$ to $\mathcal{A}$.
    \item $\mathcal{A}$ outputs a set $S^* \subseteq U$ with $|S^*| = t-1$ and messages $\{m_u\}_{u \in U}$.
    \item The challenger provides $\mathcal{A}$ the shares $\{\{s_{u,v}\}_{v \in U}\}_{u \in S^*}$ in $\{\{e_{u,v}\}_{v \in U}\}_{u \in S^*}$ and a cipher set  \[\{c_u \leftarrow Enc(parm, pk, m_u)\}_{u \in U}.\]
    \item $\mathcal{A}$ issues a polynomial number of adaptive queries of the form $(S \subseteq U, \{c_u\}_{u \in U^*}, \{\alpha_u\}_{u \in U^*})$ where $U^* \subseteq U$. For each query, the challenger computes $\hat{c} \leftarrow Eval(parm, \{c_u\}_{u \in U^*}, \{\alpha_u\}_{u \in U^*})$ and provides $\mathcal{A}$ \[\{\hat{m}_u \leftarrow ParDec(parm, \hat{c},\{e_{u,v}\}_{v \in U})\}_{u \in S}.\]
    \item At the end of the experiment, $\mathcal{A}$ outputs a distinguishing bit $b$.
    \end{enumerate}
\item $Expt_{\mathcal{A},Ideal}(1^\lambda)$:
\begin{enumerate}
    \item The adversary outputs $U$ and $V$ where $|U| =n$ and $|V| = t$ specify an access structure.
    \item The challenger runs \[(parm, pk, \{\{e_{v,u}\}_{v \in U}\}_{u \in U}, st) \leftarrow \mathcal{C}_1(1^\lambda,U,t)\] and provides $(parm, pk, \{\{e_{v,u}\}_{v \in U}\}_{u \in U})$ to $\mathcal{A}$.
    \item $\mathcal{A}$ outputs a set $S^* \subseteq U$ with $|S^*| = t-1$ and messages $\{m_u\}_{u \in U}$.
    \item The challenger provides $\mathcal{A}$ shares $\{\{s_{u,v}\}_{v \in U}\}_{u \in S^*}$ and ciphers \[\{c_u \leftarrow Enc(parm, pk, m_u)\}_{u \in U}.\]
    \item $\mathcal{A}$ issues a polynomial number of adaptive queries of the form $(S \subseteq U, \{c_u\}_{u \in U^*}, \{\alpha_u\}_{u \in U^*})$ where $U^* \subseteq U$.  For each query, the challenger runs
    \[{\{ {\hat m_u}\} _{u \in S}} \leftarrow {\mathcal{C}_2}\left( \begin{gathered}
  S,{\{ {c_u}\} _{u \in {U^*}}},{\{ {\alpha _u}\} _{u \in {U^*}}}, \hfill \\
  {\{ {m_u}\} _{u \in U}},{\{ {c_u}\} _{u \in U}},st \hfill \\
\end{gathered}  \right)\]
    and provides $\{\hat{m}_u\}_{u \in S}$ to $\mathcal{A}$.
    \item At the end of the experiment, $\mathcal{A}$ outputs a distinguishing bit $b$.
    \end{enumerate}
\end{itemize}
\end{definition}

\section{Secure Linear Aggregation Protocol}
We provides two protocols in this section. The first is a basic protocol showing how to embed a DTAHE scheme to the secure aggregation framework in \cite{B17}. The second is the secure linear aggregation protocol against an active adversary.

\subsection{A Basic Protocol}
As shown in Fig. \ref{pro1}, initially a server $S_1$ runs $parm \leftarrow Setup(1^\lambda)$ to provide a system-wide parameters for all users, and each user runs $(sk_u,pk_u)\leftarrow KeyGen(parm)$ to produce their key pairs. The server and users then runs a four-round protocol $\pi$ with an agreed threshold value $t$. We next focus on a period $e$ to show the protocol.
\begin{figure}[htbp]
  \centering
  \includegraphics[width=0.35\textwidth]{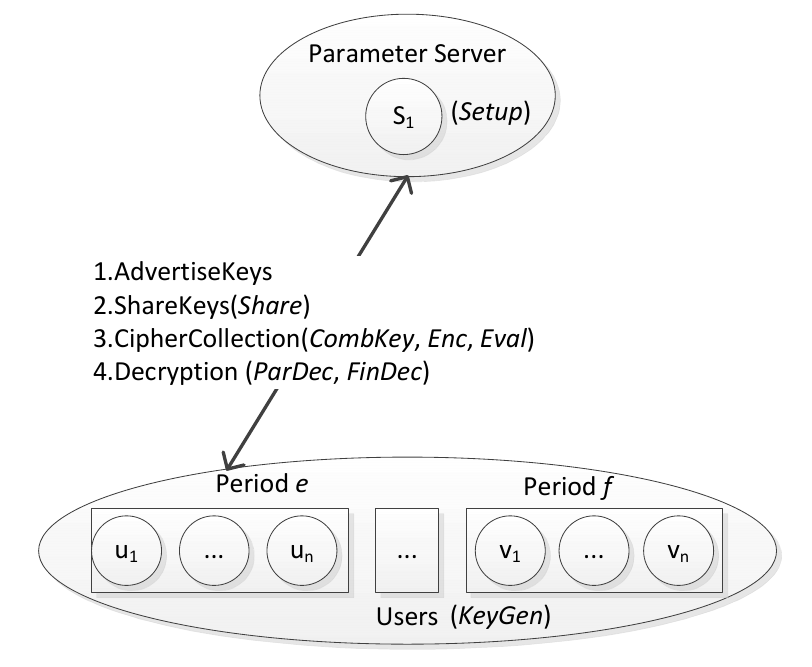}\\
  \caption{Basic Aggregation Protocol from DTAHE}\label{pro1}
\end{figure}
\begin{enumerate}
\item Round 1 (AdvertiseKeys). When a user $u$ is active, it packs a message as $m_{u,1} = (u, pk_{u})$ and sends the message to the server $S_1$. The user then waits for the first response from $S_1$ or stops.

    The server $S_1$ makes a set $U_e^1 = \emptyset$. After it receives a message $m_{u,1}=(u, pk_{u})$, it sets $U_e^1 = U_e^1 \cup \{u\}$. When $|U_e^1| > t$, $S_1$ packs a message as $m_{S,1}=\{m_{u,1}\}_{u \in U_e^1}$. The message $m_{S,1}$ is broadcasted to all the users in $U_e^1$ as their responses.
\item Round 2 (ShareKeys). When a user $u$ receives $m_{S,1}$, it makes a $U_e^1$ set from $m_{S,1}$ and checks $|U_e^1| \geq t$. If the verification passes, the user executes as follows:
    \begin{enumerate}
    \item It makes a public key set $\{pk_v\}_{v \in U_e^1\backslash\{u\}}$ from $m_{S,1}$ and produce encrypted shares by
    \[{\{ {e_{v,u}}\} _{v \in U_e^1}} \leftarrow Share\left( \begin{gathered}
  parm,{\{ p{k_v}\} _{v \in U_e^1\backslash \{ u\} }}, \hfill \\
  t,s{k_u} \hfill \\
\end{gathered}  \right).\]
    \item It packs $m_{u,2} = (u, \{(v, e_{v,u})\}_{v\in U_e^1 \backslash \{u\}})$ and sends the message to $S_1$. Then $u$ waits for the second response or stops.
    \end{enumerate}

    The server $S_1$ builds a set $U_e^2$ in the same way as $U_e^1$ satisfying $U_e^2 \subseteq U_e^1$. When $|U_e^2| > t$, $S_1$ packs a message $m_{S,2,v} = (\{(u, e_{v,u})\}_{u \in U_e^2\backslash \{v\}})$ for each user $v \in U_e^2$, and sends them to the users in $U_e^2$.

\item Round 3 (CipherCollection).  When a user $u$ receives a message $m_{S,2,u}$, it makes a set $U_e^2$ from $m_{S,2,u}$. If $|U_e^2| \geq t$, $u$ executes as follows:
    \begin{itemize}
    \item It runs $pk \leftarrow CombKey(parm, \{pk_{u}\}_{u \in U_e^2})$.
    \item It runs $c_u \leftarrow Enc(parm, pk, m_u)$ to compute the cipher of its private input $m_u$.
    \item It packs a message $m_{u,3}=(u,c_u)$ and sends the message to $S_1$. Then $u$ waits for the third response or stops.
    \end{itemize}

    The server $S_1$ builds a set $U_e^3$ in the same way as $U_e^1$ satisfying $U_e^3 \subseteq U_e^2$. It runs \[\hat{c} \leftarrow Eval(parm, \{c_u\}_{u \in U_e^3}, \{\alpha_u\}_{u \in U_e^3})\] and sends $m_{S,3} = \hat{c}$ to users in $U_e^3$.

\item Round 4 (Decryption). When a user $u$ receives $m_{S,3}$, it runs \[\hat{m}_u \leftarrow ParDec(parm,\hat{c},\{e_{u,v}\}_{v \in U_e^2}).\] $u$ then packs a message $m_{u,4}= (u,\hat{m}_u)$ and sends the message to the server $S_1$. The user $u$ finishes the client protocol $\pi_u$ for the period $e$ and stops.

    The server $S_1$ builds a set $U_e^4$ in the same way as $U_e^1$ satisfying $U_e^4 \subseteq U_e^3$. If $|U_e^4| \geq t$, it runs \[m \leftarrow FinDec(parm,\hat{c},\{\hat{m}_u\}_{u \in U_e^4})\] to get $m$ as an aggregated value. The server finishes the server side protocol $\pi_S$ for the period $e$ and waits for the next period.
\end{enumerate}

\begin{remark}
The basic protocol uses the framework in \cite{B17} to tolerate user dropping out. It needs the sever to be honest. A malicious server may simply return a cipher $c_u$ as $m_{S,3}$ to decrypt the cipher.
\end{remark}

\subsection{The Secure Linear Aggregation Protocol}
Bonawitz et al. \cite{B17} use a certificate authority (CA) and an extra consistency round to defend against an active adversary. We also use a CA but keep our protocol four rounds. We introduce a smart contract in a blockchain to encourage the server $S_1$ and users to be honest.

The blockchain is described as $\vec{\sigma}_{\lambda_t+1} = \Upsilon(\vec{\sigma}_{\lambda_t}, Tx)$ \cite{Ethereum} where $\Upsilon$ is a state transition function, $\vec{\sigma}_{\lambda_t}$ is the system state in a block height $\lambda_t$, and $Tx$ is a transaction in the system. A transaction could be described as \[Tx = (from, to, value, data, aux, sig)\] where $from$ and $to$ fields are the sender and receiver accounts of the transaction, $value$ field is the values to be transferred from the sender to the receiver, $data$ field is arbitrary data of the sender, $aux$ field is the auxiliary information used in the blockchain and $sig$ field is the sender's signature. The instance of the blockchain certainly could be the Ethereum \cite{Ethereum}. Any other blockchain system supporting state transition is fine.

Initially, the server $S_1$ has an account $acc_S$ and a user $u$ has an account $acc_u$ in the blockchain. The server should deploy a smart contract $EncCheck$ on the blockchain which will have an account $acc_E$ after deployment.
\begin{definition}
($EncCheck$). The smart contract includes three functions. A function $Init$ is for the server $S_1$ to deposit values and set parameters. A function $Record$ is for a user to record their encrypted inputs. A function $Check$ is for the server $S_1$ and users to check the correctness of a protocol transcript.
\begin{itemize}
\item $Init$: If the transaction $Tx$ is signed by the server $S_1$, the $data$ field of the transaction includes ``$Init$'' as the function name, and a threshold value $t$ and the number of periods $prds$ supported by the contract as the arguments, the function checks a deposit variable $Dep$ of the server $S_1$. If $Dep < MinValue*prds$, it checks that $Dep+value\geq MinValue*prds$ where $MinValue$ is the smallest deposit value for a period. If the check fails, it stops. Otherwise, it transfers values from the server account to the contract account, and stores $(acc_S, t, Dep, prds)$ in the contract.
\item $Record$: If the $data$ field of the transaction includes ``$Record$'' as the function name, and an epoch session id $esid_u$ and a cipher $c_{acc_u}$ as the arguments, the function stores $(esid_u, acc_u, c_{acc_u})$ in the contract.
\item $Check$: If the transaction $Tx$ is singed by the server $S_1$, the $data$ field of the transaction includes ``$Check$'' as the function name, and a termination flag $draw$, an epoch session id $esid_S$, a cipher $c_S$, a list of user accounts $L_A$, and coefficients $\{\alpha_{acc_u}\}_{acc_u \in L_A}$ as the arguments, the function checks whether $esid_S$ is new. If the $esid_S$ is used before in another transaction $Tx^\prime$, it checks whether the two transactions are the same. If they are different, the check fails. If they are the same, the function stops. The function forms a set \[A_A = \{acc_u: acc_u \in LA \wedge esid_u = esid_S\}.\]  If $esid_S$ is new, it checks that the size $|A_A| \geq t$, and \[c_S = Eval(parm, \{c_{acc_u}\}_{acc_u \in A_A},\{\alpha_{acc_u}\}_{acc_u \in A_A}).\] If any check fails, the value $MinValue$ is shared by users in the $A_A$ set. Otherwise, it updates $prds=prds-1$, and if $draw$ is $true$, it transfers the deposit back to the account of the server after six new blocks. When the deposit is cleared, the contract suicides.
\end{itemize}
\end{definition}

Now the server $S_1$ has two tasks to initialize the protocol. At first, the server $S_1$ runs $parm \leftarrow Setup(1^\lambda)$ to provide system-wide parameters for all users. Secondly, the server $S_1$ produces a transaction as \[Tx_{S,1} = (acc_S,acc_E,value,Init||t||prds,aux,sig)\] to initialize the deployed smart contract.

A user $u$ has three tasks to participate in the protocol II. At first, it runs $(sk_u, pk_u) \leftarrow KeyGen(parm)$ to produce protocol keys. Secondly, it produces a transaction \[Tx_{u,1} = (acc_u,acc_u,0,pk_u,aux,sig)\] and sends the transaction to the blockchain to store its public keys. Thirdly, the user $u$ should produce a signature key pair $(sk_{Sigu}, vk_{Sigu})\leftarrow SIG.Gen(1^\lambda)$ using a signature scheme $(SIG.Gen,SIG.Sign,SIG.Ver)$ that is existential unforgeable against adaptive chosen messages (EUF-CMA). And then $u$ applies for a certificate $cert_u$ of the verifying key $vk_{Sigu}$ from a CA.

The CA and miners are auxiliary entities in the protocol. A CA is used to defend against a ``Sybil'' attack where an adversary may create many blockchain accounts to act as users. The miners of a blockchain receive transactions, pack them into blocks, reach a consensus on a block and update the global state. The public keys of users, parameters of the server $S_1$, and deposited values of the server $S_1$ are states of the blockchain maintained by the miners. The whole picture of the protocol is in Fig. \ref{p2}. We next focus on a period $e$ to show the protocol.
\begin{figure}[htbp]
  \centering
  \includegraphics[width=0.45\textwidth]{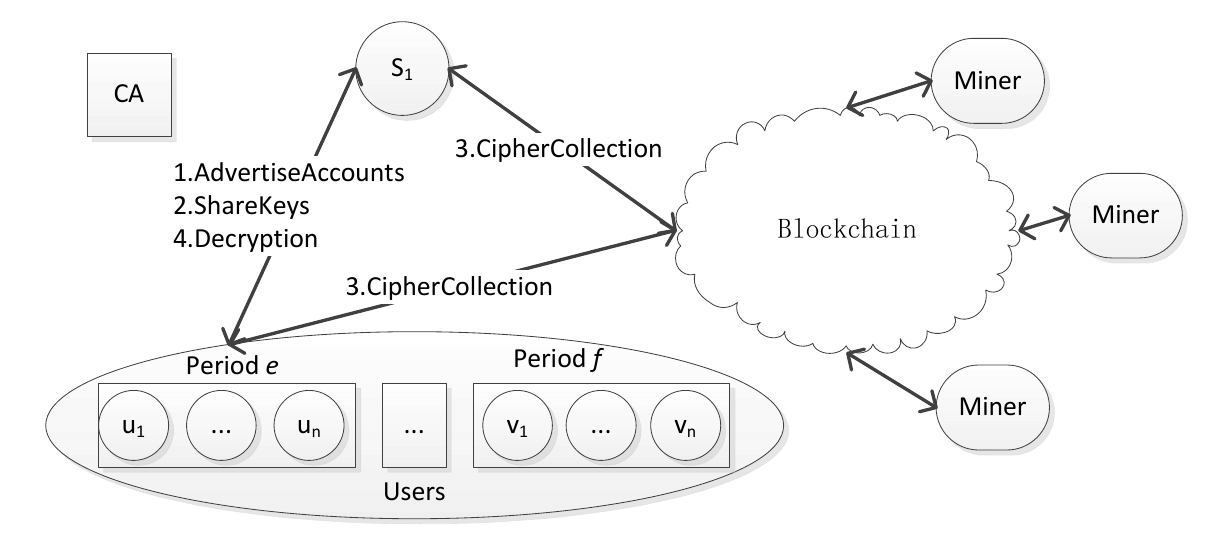}\\
  \caption{The Secure Linear Aggregation Protocol}\label{p2}
\end{figure}

\begin{enumerate}
\item Round 1 (AdvertiseAccounts). When a user $u$ is active, it looks up the deposit $Dep$, supported periods $prds$ and the threshold $t$ in the smart contract account $acc_E$ indexed by the server account $acc_S$. If $Dep$, $t$, or $prds$ are too small for the user's local policy, the user simply stops. Otherwise, it produces a signature \[\sigma_u \leftarrow SIG.Sign(sk_{Sigu}, acc_u||T_u)\] where $T_u$ is a time stamp of the user. It packs a message as $m_{u,1} = (u, acc_u, T_u, \sigma_u, cert_u)$ and sends the message to the server $S_1$. The user $u$ then waits for the first response from $S_1$ or stops.

    The server $S_1$ builds a set $U_e^1$ as in the basic protocol. It then packs a message as $m_{S,1}=(\{m_{u,1}\}_{u \in U_e^1},esid)$ where $esid$ is a session number for the period. The message $m_{S,1}$ is broadcasted to all the users in $U_e^1$ set as their response.
\item Round 2 (ShareKeys). When a user $u$ receives $m_{S,1}$, it parses $m_{S,1}$ to extract the identities of users which form a set $U_e^1$. If checks that $|U_e^1| \geq t$, and the time stamps, signatures and certificates are correct. If the verifications pass, $u$ looks up the blockchain to find public keys of users in $U_e^1$ by their accounts which form a set $\{pk_v\}_{v \in U_e^1 \backslash \{u\}}$, and executes as in the basic protocol.

    The server $S_1$ executes in the same way as in the basic protocol to produce $m_{S,2,v}$ for a user $v$.

\item Round 3 (CipherCollection).  When a user $u$ receives a message $m_{S,2,u}$, it makes a set $U_e^2$ from $m_{S,2,u}$. If $|U_e^2| \geq t$, it executes in the same way as in the basic protocol to compute $pk$ and $c_u$. Then $u$ creates a transaction \[Tx_{u,2}=(acc_u,acc_E,0,Record||esid||c_u,aux,sig)\] and sends the transaction to the blockchain. The user $u$ then waits for a response transaction of the server $S_1$ from the blockchain network or stops.

    The server $S_1$ makes a set $U_e^3 = \emptyset$. After it receives a transaction $Tx_{u,2}$, if the transaction includes the same $esid$ as the current session number, $S_1$ sets $U_e^3 = U_e^3 \cup \{u\}$ where the identity $u$ is indexed by the $acc_u$. It then runs \[\hat{c} \leftarrow Eval(parm, \{c_u\}_{u \in U_e^3}, \{\alpha_u\}_{u \in U_e^3}),\] and packs a transaction
    \[T{x_{S,2}} = \left( \begin{gathered}
  ac{c_S},ac{c_E},0,Check||false|| \hfill \\
  esid||\hat c||{\{ ac{c_u}\} _{u \in U_e^3}},aux,sig \hfill \\
\end{gathered}  \right)\]
    and sends the transaction to the blockchain.

\item Round 4 (Decryption). When a user $u$ receives $Tx_{S,2}$, it builds a $U_e^3$ set indexed by $\{acc_u\}_{u \in U_e^3}$. If $U_e^3 \subseteq U_e^2$, $|U_e^3| \geq t$, and the $esid$ field is the current session number, $u$ then executes in the same way as in the basic protocol to interact with the server $S_1$.

    The server $S_1$ executes in the same way as in the basic protocol to get an aggregated result.
\end{enumerate}

When the FL process is to stop, $S_1$ may send a $Tx_{S,2}$ with the $true$ termination flag. The $EncCheck$ contract will delay the withdraw action until six blocks are generated. This is intentionally for users to check the correctness of the server. Before the deposit is given back to the server, a user could send $Tx_{S,2}$ received in an period to the contract. If there are different $Tx_{S,2}$ transactions with the same $esid$ field, the deposit will be shared by users in that period.

\begin{remark}
In a period, two transactions are transferred by the blockchain network. Miners, the parameter server and users will receive transactions, parse them, and process them independently.
\end{remark}

\subsection{Security Analysis}
We prove the security of the secure linear aggregation protocol against an active adversary. In a high level, we divide the global outputs into four views, and analyze the advantages of the adversary as a new view appears.

\begin{theorem}
Suppose that a DTAHE scheme is sematic and simulation secure, a signature scheme $SIG$ is EUF-CMA secure, and the secret sharing scheme in the DTAHE has a privacy property \cite{BGGJK}. Let $t$ be the threshold value, $n_c$ be the maximal number of corrupted users in a period, and $\lambda_m$ be the minimal number of private inputs to make the final aggregated value hide one input. For any uncorrupted user $u$, against an adversary $\mathcal{A}$, if $\lambda_m \geq t - n_c$ and the deposit of the server $S_1$ in the smart contract $EncCheck$ has not lost, the protocol satisfies $GO_{\pi, e, \{S_1\}\cup U_e, \mathcal{A}}(\vec{x}) \approx GO_{\pi, e, \{S_1\}\cup U_e, \mathcal{A}}(\vec{x}^\prime)$.
\end{theorem}

{\noindent \it{Proof}}:
A global output $GO_{\pi,e, \{S_1\}\cup U_e, \mathcal{A}}(\vec{x})$ consists of four views $\{view_i\}_{i \in \{1,...4\}}$. The $view_1$ includes the outputs of the initialization procedures and the first round of the server $S_1$ and users. The views $\{view_2,view_3,view_4\}$ include the outputs of their corresponding rounds. We consider two global outputs $GO$ and $GO^\prime$ with inputs $\vec{x}$ and $\vec{x}^\prime$. Note that the only difference of the two inputs are the inputs of the user $u$. We denote the four views of $GO^\prime$ as $\{view_i^\prime\}_{i \in \{1,...4\}}$.
\begin{itemize}
\item Initially, $view_1$ and $view_1^\prime$ include the same certificates, public keys and accounts in the initialization procedure of the protocol $\pi$.
\item When $\pi$ executes, $view_1$ and $view_1^\prime$ include message-signature pairs of users and the first response of the server. $view_2$ and $view_2^\prime$ include encrypted shares and the second response of the server. Since private inputs of the user $u$ are not involved, the two global outputs are indistinguishable until now.
\item The $view_{3}$ and $view_3^\prime$ consist of two kinds of transactions, separately. Transactions of users in $U_e^3$ could form a ciphertext set $\{c_u\}_{u \in U_e^3}$. The transaction of the server $S_1$ includes $\hat{c}$. In the $view_3$, the ciphertext of the target user $u$ is $c_u \leftarrow Enc(parm, pk, m_u)$, and the evaluated ciphertext of the server $S_1$ is $\hat{c} \leftarrow Eval(parm, \{c_u\}_{u \in U_e^3}, \{\alpha_u\}_{u \in U_e^3})$. In the $view_3^\prime$, the ciphertext of the target user $u$ is $c_u \leftarrow Enc(parm, pk, 0)$, and the evaluated ciphertext of the server is produced in the same way. Now the chances of $\mathcal{A}$ increases.
    \begin{itemize}
    \item $\mathcal{A}$ could exploit the differences of the two views $view_3$ and $view_3^\prime$. Since the only difference is the message $m_u$ and the message $0$, if the DTAHE is sematic secure, $\mathcal{A}$ has only negligible advantages to distinguish the two views.
    \item $\mathcal{A}$ could exploit the combined views of $view_3 \cup view_2$ and $view_3^\prime \cup view_2^\prime$. Suppose that $\mathcal{A}$ could get the shared secrets in $view_2$ and $view_2^\prime$. Then $\mathcal{A}$ could decrypt $c_u$ by the $ParDec$ and $FinDec$ algorithms. However, since the DTAHE is simulation secure, the secret shares are protected well. $\mathcal{A}$ has only negligible advantages to distinguish the two combined views.
    \item $\mathcal{A}$ could exploit the combined views of $view_3 \cup view_2 \cup view_1$ and $view_3^\prime \cup view_2^\prime \cup view_1^\prime$ in the two global outputs.

        At first, Suppose that $\mathcal{A}$ corrupts users in $C^*$ and $|C^*| = n_c$. It then knows secret keys $\{sk_{u}\}_{u \in C^*}$ which could be used to decrypt ciphers in $view_2$. The decrypted shares may help $\mathcal{A}$ to decrypt $c_u$. However, since we assume $n_c < t$ and the secret sharing scheme has a privacy property\cite{BGGJK}, this event happens negligibly.

        Secondly, $\mathcal{A}$ may make special $view_1$ and $view_1^\prime$. $\mathcal{A}$ could establish $t$ blockchain accounts and store public keys in each account. Then $\mathcal{A}$ acts as users to cheat the target user $u$. However, since messages in $view_1$ and $view_1^\prime$ are signed by users who have certificates from a CA, the accounts of $\mathcal{A}$ could not be identified by the user $u$. And since we explicitly include a time stamp in the first message of a user, $\mathcal{A}$ could not replay messages of users in other periods. To satisfy the claims, the signature scheme $SIG$ should be EUF-CMA secure.
    \end{itemize}
\item $view_{4}$ and $view_4^\prime$ consists of partially decrypted ciphers and the final aggregated value. $\mathcal{A}$ has more chances:
    \begin{itemize}
    \item $\mathcal{A}$ could exploit the partially decrypted ciphers to discover secret shares. Since the DTAHE is simulation secure, $\mathcal{A}$ has only negligible advantages.
    \item $\mathcal{A}$ could exploit the final aggregated value. An obvious strategy is that in the transaction of the server $Tx_{S,2}$, only $c_u$ is included. To make the strategy more practical, $\mathcal{A}$ may aggregate some ciphers of corrupted users and the target user. Then the final output is the sum of $n_c+1$ inputs.

        The strategy is a possible way to get inputs of a user at the cost of the deposit in the smart contract. Since $Tx_{S,2}$ comes from the blockchain network, miners will execute the $Check$ function in the $EncCheck$ smart contract. The function takes $\{acc_u\}_{u \in U_e^3}$ as $L_A$, and checks the correctness of the evaluation of the server $S_1$. The smart contract builds a $A_A \subseteq L_A$ and checks that $|A_A| \geq t$. So if $n_c+1 < t$, the check fails and the server $S_1$ will be punished. To avoid the punishment, $\mathcal{A}$ should prepare a $L_A$ set with $|L_A|\geq n_c + \lambda_m \geq t$.

        $\mathcal{A}$ may try to withdraw their deposit before the penalty is paid. However, the smart contract needs extra six blocks to confirm the request of the server. The six-block waiting time is the last chance of users to check the behaviours of the server. Users are encouraged to send the received transaction of the server to the blockchain since they are the possible beneficiaries.
    \end{itemize}
\end{itemize}
In summary, if the cryptographic primitives are secure, $\lambda_m \geq t - n_c$, and the server pays no penalties, the adversary has only negligible advantage to distinguish the two outputs $GO$ and $GO^\prime$. $\hfill$

\begin{remark}
We explicitly introduce a parameter $\lambda_m$ to describe the number of private inputs of uncorrupted users in a period. The parameter is in fact used in \cite{B17} implicitly. The consistency check round \cite{B17} makes sure that at least $t$ private inputs are summed. If $n_c$ inputs are known by an adversary $\mathcal{A}$, then in their protocol $\lambda_m \geq t - n_c$.
\end{remark}

\section{A DTAHE Instance}
Since all known construction methods of DTAHE schemes have some weaknesses. We provide a new construction method to produce a lattice based DTAHE scheme for an interactive protocol.

\subsection{The Instance}
We use the BFV scheme \cite{FV12,B12} as a basic building block which is implemented in the SEAL library \cite{seal}. We need a sematic secure hybrid encryption scheme \cite{HHK10} $HPKE=(HPKE.Gen,HPKE.Enc,HPKE.Ver)$ to encrypt shares. A lattice based $HPKE$ instance will give us a fully lattice DHAHE instance that may be post quantum secure. We also need the Shamir secret sharing scheme denoted by $SS = (SS.Split, SS.Recover)$.

\begin{enumerate}
\item $Setup(1^\lambda) \rightarrow parm$: It takes as input a security parameter $\lambda$, produces a parameter set $parm = (d,f(x),h,R,R_h,\chi,\mu,\vec{a},l,\lambda)$, where $d$ is the degree depending on $\lambda$ of a cyclotomic polynomial $f(x)$, $h\geq 2$ is an integer depending on $\lambda$, $R$ is a ring $R = \mathbb{Z}[x]/(f(x))$, $R_h$ is the set of polynomials in $R$ with coefficients in $\mathbb{Z}_h$, $\chi$ here is in fact defined as discrete Gaussian distribution, $\mu$ is a uniform distribution, $\vec{a}$ is uniformly selected from $R_h$ as $\vec{a} \leftarrow R_h$, and $l$ is an integer depending on $\lambda$.
\item $KeyGen(parm) \rightarrow (sk_u, pk_u)$: It selects $\vec{s}_u \leftarrow R_3$ and samples $\vec{e}_u \leftarrow \chi$. It sets $sk_{u,0} = \vec{s}_u$ and $pk_{u,0} = [-(\vec{a}\cdot \vec{s}_u+\vec{e}_u)]_h$. It runs $(pk_{u,1},sk_{u,1})\leftarrow HPKE.Gen(1^\lambda)$. The output is $sk_u = (sk_{u,0},sk_{u,1})$ and $pk_u=(pk_{u,0},pk_{u,1})$.
\item $Share(parm, \{pk_{v}\}_{v \in U \backslash\{u\}}, t, sk_{u}) \rightarrow \{e_{v,u}\}_{v \in U \backslash \{u\}}$: It samples $\vec{e}_{u,2} \leftarrow \chi$, computes $n = |\{pk_{v}\}_{v \in U \backslash\{u\}}|+1$, sets $ss_u = \{sk_{u,0},\vec{e}_{u,2}\}$, and for each coefficient $ss_{u,i}$ of the elements in $ss_u$, computes $\{\vec{s}_{v,u,i}\}_{v \in U} \leftarrow SS.Split(ss_{u,i}, n, t, h)$. For all the shares to $v \in U \backslash \{u\}$, it computes a cipher $e_{v,u}=HPKE.Enc(pk_{v,1},\{\vec{s}_{v,u,i}\}_{i \in |ss_u|*d})$.
\item $CombKey(parm, \{pk_{u}\}_{u \in U}) \rightarrow pk$: It computes $pk=[\sum_{u \in U}pk_{u,0}]_h$.
\item $Enc(parm, pk, m_u) \rightarrow c_u$: It selects $\vec{u}_u \leftarrow R_3$, $\vec{e}_{u,0},\vec{e}_{u,1} \leftarrow \chi$, computes $c_{u,0}=[\vec{a}\cdot\vec{u}_u+\vec{e}_{u,0}]_h$ and $c_{u,1} = [pk\cdot \vec{u}_u+\vec{e}_{u,1}+\lfloor h/l \rfloor\cdot m_u]_h$. It sets $c_u = (c_{u,0}, c_{u,1})$.
\item $Eval(parm, \{c_u\}_{u \in U},\{\alpha_u\}_{u \in U})\rightarrow \hat{c}$: It computes $\hat{c}_0= [\sum_{u \in U}\alpha_{u}c_{u,0}]_h$, $\hat{c}_1=[\sum_{u \in U}\alpha_{u}c_{u,1}]_h$ and sets $\hat{c}= (\hat{c}_0,\hat{c}_1)$.
\item $ParDec(parm,\hat{c},\{e_{u,v}\}_{v \in U})\rightarrow \hat{m}_u$: It decrypts the shares for the user $u \in U$ from the user $v \in U$ as $\vec{s}_{u,v}\leftarrow HPKC.Dec(sk_{u,1},e_{u,v})$. It parses the shares of coefficients as shares of elements in $R$, sets $(\vec{se}_{u,v},\vec{ssk}_{u,v})=\vec{s}_{u,v}$, then computes $\hat{m}_u = [\hat{c}_0 \cdot \sum_{v \in U}\vec{ssk}_{u,v}+\sum_{v \in U}\vec{se}_{u,v}]_h$.
\item $FinDec(parm,t, \hat{c},\{\hat{m}_u\}_{u \in V}) \rightarrow m$: It recovers $cs=[\sum_{u \in V}li_{u}\hat{m}_u]_h$ where $li_{u}$ is the Lagrange coefficient of the user $u$ with respect to the user set $V$, and computes the final output $m = [\lfloor\frac{l \cdot [\hat{c}_1+cs]_h}{h}\rceil]_l$.
\end{enumerate}
\begin{remark}
If the data dimension of the message $m_u$ is greater than $d$, the number of noise samples $\vec{e}_{u,2} \in R$ increases.
\end{remark}
\subsection{Security Analysis}
The security of our scheme could be reduced to a variant of the classical ring version decisional learning with errors (RLWE) problem \cite{LPR10,FV12}. The variant is named $n$-Decision-RLWE problem.
\begin{definition}
\label{simsecure}
($n$-Decision-RLWE).
For a random set $\{\vec{s}_i \in R_h\}_{i \in \{1,\ldots,n\}}$ and a distribution $\chi$ over $R$, denote with $A_{\{\vec{s}_i \in R_h\}_{i \in \{1,\ldots,n\}},\chi,\mu}$ the distribution by choosing a uniformly random element $\vec{a} \leftarrow R_h$ and $n$ noise term $\{\vec{e}_i \leftarrow \chi \}_{i \in \{1,\ldots,n\}}$ and outputting $(\vec{a},\{[\vec{a}\cdot\vec{s}_i+\vec{e}_i]_h\}_{i \in \{1,\ldots,n\}})$. The problem is then to distinguish between the distribution $A_{\{\vec{s}_i \in R_h\}_{i \in \{1,\ldots,n\}},\chi,\mu}$ and a uniform distribution $\mu$ over $R_h^{n+1}$.
\end{definition}

By a hybrid argument, one could conclude that if an adversary has an advantage at least $\epsilon_{n\text{-}RLWE}$ to solve the $n$-Decision-RLWE problem, the adversary has an advantage at least $\frac{1}{n}\epsilon_{n\text{-}RLWE}$ to solve the classical RLWE problem in \cite{LPR10,FV12}.

The first proof is about the evaluation correctness in the definition \ref{corrdef}.
\begin{theorem}
\label{corr}
Assume that $U$ is the user set, $\chi$ is $B$-bounded and the maximal infinity norm of elements in the set $\{\alpha_u\}_{u \in U}$ is $A$, the evaluation of the DTAHE is correct with probability $1$ if $|U|B(1+\delta_RA(1+2\delta_R|U|)) < \frac{h}{2l}$.
\end{theorem}

\begin{proof}
Since \[\hat{m}_u = [\hat{c}_0 \cdot \sum_{v \in U}\vec{ssk}_{u,v}+\sum_{v \in U}\vec{se}_{u,v}]_h,\] we have the equation (\ref{pareq}) due to the Lagrange interpolation.
\begin{equation}
\begin{split}\label{pareq}
 cs&= [\sum_{u \in V}li_{u}\cdot (\hat{c}_0 \cdot \sum_{v \in U}\vec{ssk}_{u,v}+\sum_{v \in U}\vec{se}_{u,v})]_h \\
  &= [\sum_{u \in V}li_{u}\cdot(\hat{c}_0 \cdot \sum_{v \in U}\vec{ssk}_{u,v})\\
  &+\sum_{u \in V}li_{u}\cdot(\sum_{v \in U}\vec{se}_{u,v})]_h\\
  &= [\hat{c}_0 \cdot (\sum_{u \in V}li_{u}\cdot(\sum_{v \in U}\vec{ssk}_{u,v}))+\sum_{u \in U}\vec{e}_{u,2}]_h\\
  &= [\hat{c}_0 \cdot \sum_{u \in U}sk_{u,0}+\sum_{u \in U}\vec{e}_{u,2}]_h\\
\end{split}
\end{equation}

Let $X = [\hat{c}_1+cs]_h$. Since \[\hat{c}_0= [\sum_{u \in U}\alpha_uc_{u,0}]_h,\] \[\hat{c}_1=[\sum_{u \in U}\alpha_uc_{u,1}]_h,\] \[c_{u,0}=[\vec{a}\cdot\vec{u}_u+\vec{e}_{u,0}]_h,\] \[c_{u1} = [pk\cdot \vec{u}_u+\vec{e}_{u,1}+\lfloor h/l \rfloor\cdot m_u]_h,\] \[pk=[\sum_{u \in U}pk_{u,0}]_h,\] and \[pk_u = [-(\vec{a}\cdot sk_{u,0}+\vec{e}_u)]_h,\] we have the equation \ref{cormid}.
\begin{equation}
\begin{split}\label{cormid}
X &= [\hat{c}_1+\hat{c}_0 \cdot \sum_{u \in U}sk_{u,0}+\sum_{u \in U}\vec{e}_{u,2}]_h \\
 &=[\sum_{u \in U}\alpha_u(pk\cdot \vec{u}_u+\vec{e}_{u,1}+\lfloor h/l \rfloor\cdot m_u)\\
 &+\sum_{u \in U}\alpha_u(\vec{a}\cdot\vec{u}_u+\vec{e}_{u,0}) \cdot (\sum_{u \in U}sk_{u,0})\\
 &+\sum_{u \in U}\vec{e}_{u,2}]_h\\
\end{split}
\end{equation}

Let \[NS_0=\sum_{u \in U}\alpha_u\vec{e}_{u,1}+(\sum_{u \in U}\alpha_u\vec{e}_{u,0}) \cdot (\sum_{u \in U}sk_{u,0})+\sum_{u \in U}\vec{e}_{u,2},\] and \[MP=\sum_{u \in U}\alpha_u\lfloor h/l \rfloor\cdot m_u.\] Let \[Y = [X-NS_0-MP]_h.\] We then have the equation \ref{corfin}.
\begin{equation}
\begin{split}\label{corfin}
Y &= [\sum_{u \in U}\alpha_u(pk\cdot \vec{u}_u)+\sum_{u \in U}\alpha_u(\vec{a}\cdot\vec{u}_u) \cdot (\sum_{u \in U}sk_{u,0})]_h \\
 &=[\sum_{u \in U}\alpha_u(\vec{u}_u \cdot (\sum_{u \in U}-(\vec{a}\cdot sk_{u0}+\vec{e}_u)))\\
 &+\vec{a}\cdot (\sum_{u \in U}\alpha_u\vec{u}_u) \cdot (\sum_{u \in U}sk_{u,0})]_h\\
 &=[-\sum_{u \in U}\alpha_u\vec{u}_u \cdot \sum_{u \in U}\vec{e}_u]_h\\
\end{split}
\end{equation}
Let $NS=NS_0-\sum_{u \in U}\alpha_u\vec{u}_u \cdot \sum_{u \in U}\vec{e}_u$, then $X=[NS+MP]_h$. We then have the equation \ref{meq}.
\begin{equation}
\begin{split}\label{meq}
m &= [\lfloor\frac{l \cdot [\hat{c}_1+cs]_h}{h}\rceil]_l\\
  &=[\lfloor\frac{l \cdot X}{h}\rceil]_l\\
  &=[\lfloor\frac{l \cdot (NS+\sum_{u \in U}\alpha_u\lfloor h/l \rfloor\cdot m_u)}{h}\rceil]_l\\
  &=[\sum_{u \in U}\alpha_um_u]_l+[\lfloor\frac{l \cdot NS}{h}\rceil]_l\\
\end{split}
\end{equation}

If $NS < \frac{h}{2l}$, the above decryption is correct. Since $\vec{u}_u, sk_{u0} \leftarrow R_3$, $\vec{e}_{u0}$, $\vec{e}_{u1}$, $\vec{e}_u$, $\vec{e}_{u2} \leftarrow \chi$, the maximal infinity norm of elements in the set $\{\alpha_u\}_{u \in U}$ is $A$, the infinity norm of $NS$ is
\begin{equation}
\begin{split}
||NS|| &\leq \delta_RA|U|B+\delta_R^2A|U|^2B+|U|B+\delta_R^2A|U|^2B\\
  &=|U|B(1+\delta_RA(1+2\delta_R|U|))\\
\end{split}
\end{equation}
\end{proof}

The second proof is for the privacy of messages in the definition \ref{semdef}.
\begin{theorem}
If there is an adversary $\mathcal{A}$ with advantage $\epsilon_{sem}$ to make the experiment $Expt_{\mathcal{A},Sem}(1^\lambda)$ output $1$, one could construct a challenger to break the $n$-decision-RLWE problem with an advantage $\frac{1}{2}\epsilon_{sem}$ under the condition that the secret sharing scheme $SS$ has the privacy property \cite{BGGJK} and the hybrid encryption scheme $HPKE$ is sematic secure.
\end{theorem}
\begin{proof}
With $|U|$ and $t$, the challenger samples a $|U|$-decision-RLWE instance $(x_0,\{x_u\}_{u \in U})$. It embeds the problem instance into the DTAHE instance as follows: \[parm \leftarrow Setup(1^\lambda),\]\[parm=parm\backslash\{\vec{a}\}\cup\{x_0\},\] \[(sk_u,pk_u)\leftarrow KeyGen(parm),\] \[sk_{u,0}=0; pk_{u,0}=x_u,\]\[\{e_{v,u}\}_{v \in U\backslash\{u\}} \leftarrow Share(parm, \{pk_v\}_{v \in U \backslash\{u\}}, t, sk_u),\] \[pk \leftarrow CombKey(parm,\{pk_u\}_{u \in U}).\] It then provides $(parm, pk, \{\{e_{v,u}\}_{v \in U\backslash\{u\}}\}_{u \in U})$ to $\mathcal{A}$.

The challenger plays with $\mathcal{A}$ by $\{sk_{u,1}\}_{u \in U}$ until $\mathcal{A}$ outputs $b^\prime$.

If the $HPKE$ scheme is sematic secure, the ciphers $\{e_{v,u}\}_{v \in U\backslash\{u\}}$ leak nothing about shares. Then from the privacy property of the $SS$ scheme, if $|S|<t$, $\mathcal{A}$ can not distinguish a secret $sk_{u,0}$ from zero. So $\mathcal{A}$ should produce an educated guess $b^\prime$.

The strategy of the challenger is to use the guess of $\mathcal{A}$. If the $Expt_{\mathcal{A},Sem}(1^\lambda)$ outputs $0$, the challenger believes that the $|U|$-decision-RLWE instance is a uniform random sample from $R_h^{n+1}$.

When the input is indeed a uniform random sample from $R_h^{n+1}$, the advantage of $\mathcal{A}$ is simple negligible since the messages are masked by random values. Otherwise, the adversary has an advantage $\epsilon_{sem}$ by assumption. So the advantage of the challenger is $\frac{1}{2}\epsilon_{sem}$.
\end{proof}

The third proof is for the privacy of secret keys and shares in the definition \ref{simdef}.
\begin{theorem}
If the secret sharing scheme $SS$ has the privacy property \cite{BGGJK} and the hybrid encryption scheme $HPKE$ is sematic secure, the adversary $\mathcal{A}$ has negligible advantage to distinguish the two experiments $Exp_{\mathcal{A},Real}(1^\lambda)$ and $Expt_{\mathcal{A},Ideal}(1^\lambda)$.
\end{theorem}
\begin{proof}
The proof needs a serial of hybrid experiments between an adversary $\mathcal{A}$ and a challenger.
\begin{itemize}
\item $H_0$: This is the experiment $Exp_{\mathcal{A},Real}(1^\lambda)$ in the definition \ref{simdef}.
\item $H_1$:Same as $H_0$, except that the challenger simulates the $ParDec$ algorithm to produce $\hat{m}_u$ for queries of $\mathcal{A}$. Note that $\mathcal{A}$ has given the challenger a set $S^*$ with the size $|S^*| = t-1$. From $S^*$, the challenger could construct a set $S_C=S\backslash S^*$. For each party $u \in S_C$, $|S^* \cup \{u\}|  = t$. The challenger sets $\hat{m}_u$ as \[\tilde{m}_u = [li_u^{-1}(\lfloor{h/l}\rfloor\sum_{v  \in U}\alpha_vm_v +NS - \hat{c}_1-\sum_{v\in S^*}li_v\hat{m}_v)]_h\] where $NS$ is defined in the theorem \ref{corr}. If $u \in S^*$, the challenger computes $\hat{m}_u$ as in the game $H_0$.

    The correctness of the simulation is obviously since
    \begin{equation}
    \begin{split}\label{dsk}
    \tilde{m}_u &= [li_u^{-1}(MP +NS - (\hat{c}_1+\sum_{v\in S^*}li_v\hat{m}_v))]_h\\
    &=[li_u^{-1}(MP +NS - (\hat{c}_1+cs-li_u\hat{m}_u))]_h \\
    &= [li_u^{-1}(MP +NS - X +li_u\hat{m}_u)]_h\\
    &=\hat{m}_u \\
    \end{split}
    \end{equation}

\item $H_2$: Same as $H_1$, except that the challenger shares zero as \[\{e_{v,u}\}_{v \in U\backslash\{u\}} \leftarrow Share(parm, \{pk_v\}_{v \in U \backslash\{u\}}, t, 0).\] By the privacy property \cite{BGGJK} of the $SS$ scheme and the sematic security of the $HPKE$ scheme, $H_2$ and $H_1$ are indistinguishable.
\item $H_3$: Same as $H_2$, except that $NS$ is replaced by $\tilde{NS}$ as
    \begin{equation}
    \begin{split}
    \tilde{NS} &= NS - (\sum_{u \in U}sk_{u,0})(\sum_{u\in U}\alpha_u\vec{e}_{u,0}) \\
    & + (\sum_{u \in U}\vec{u}_{u,0})(\sum_{u\in U}\alpha_u\vec{e}_{u,0})\\
    \end{split}
    \end{equation}
    where $\vec{u}_{u,0}\leftarrow R_3$.

    Since $sk_{u,0} \leftarrow R_3$, $H_3$ and $H_2$ have the same distribution. In fact, $\tilde{NS}$ may appear in an experiment when the $\{\vec{u}_{u,0}\}_{u \in U}$ happens to be part of the secret keys of users. Now the challenger does not use the private keys of users $\{sk_u\}_{u \in U}$ or secret shares of users in $U$. So the ideal experiment $Expt_{\mathcal{A},Ideal}(1^\lambda)$ could be simulated indistinguishably.
\end{itemize}
\end{proof}

\section{Performance}
We evaluate the communication and computation costs of the secure linear aggregation protocol.
\subsection{Communication}
We have stated that there are mainly three methods \cite{P91,BD10,BGGJK} to construct a DTAHE. We concrete the method in \cite{P91} based on an elliptic curve version ElGamal (EC-ElGamal) scheme, and other methods based on the BFV \cite{B12, FC12} scheme. The user side communication overhead of the secure linear aggregation protocol is calculated as follows:
\begin{equation}
\begin{split}\label{commeq}
\sum_{i \in \{1,\ldots,4\}}m_{u,i} &= (n+2)|u|+(n-1)|e_{v,u}|+|c_u|\\
 &+|\hat{m}_u|+3|acc_u|+|T_u|+|\sigma_u|+|Sig|\\
 &+|cert_u|+|aux|+|Record|+|esid|
\end{split}
\end{equation}
where $n$ is the number of users. Table I shows the main components in the equation (\ref{commeq}). The security parameter $\lambda$ is $128$. An element in EC-ElGamal takes $33$ bytes where one byte is for $y$-coordinate. $LR$ denotes the size of a ring element in $R_h$, $SN$ the number of shares to each user and $LN$ the number of ciphers to encrypt the user input $m_u$. $LR^\prime$ and $LN^\prime$ have the same meaning as $LR$ and $LN$ with $LR \neq LR^\prime$ and $LN \neq LN^\prime$. The method to use Shamir secret sharing in  \cite{BGGJK} is denoted as BGGJK-2, and the other is denoted as BGGJK-1.
\begin{table}[!t]%%!t
\label{tabcomm}
\renewcommand{\arraystretch}{1.3}
\centering
\setlength{\abovecaptionskip}{0pt}%
\setlength{\belowcaptionskip}{10pt}%
\caption{Communication Overheads of Main Components in the Protocol}
\resizebox{0.49\textwidth}{15mm}{
\begin{tabular}{|c|c|c|c|}
\hline
                  & $|e_{v,u}|$             &  $|c_u|$                 &  $|\hat{m}_u|$  \\  \hline
Pedersen\cite{P91}        & $33+32$                 & $66*|m_u|$               & $33*|m_u|$ \\  \hline
BD\cite{BD10}       & $33+LR*SN$              & $2*LR*LN$                &$LR*LN*SN$\\  \hline  %RL==Length of a ring lement module h;SN==share of numbers n-1,t-1
BGGJK-1\cite{BGGJK}    & $33+LR*n^4$             & $2*LR*LN$                &$LR*LN*n^4$\\  \hline
BGGJK-2\cite{BGGJK}    & $33+LR^\prime$          & $2*LR^\prime*LN^\prime$  &$LR^\prime*LN^\prime$\\  \hline
Ours              & $33+LR(1+LN)$           & $2*LR*LN$                &$LR*LN$\\  \hline   %%length of noise
\end{tabular}
}
\end{table}

Fig. \ref{fig_comm} shows the main communication overhead of the protocol on user side with different DTAHE constructions when the number of user increases. We mainly consider the components in Table I. We set $|m_u|=10^5$ for all instances, and set $d=2048$ and $|h|=54$ so that $LR =d*|h|/8= 13824$ bytes and $LN=\lceil |m_u|/d \rceil = 49$. We set $t = \lceil n*2/3 \rceil$ and compute $SN= \left( \begin{array}{l}n-1\\t-1\end{array} \right)$. The values of $LR^\prime$ and $LN^\prime$ are relative to the number of users since the noise element in the $ParDec$ algorithm should be multiplied by $(n!)^2$, which are limited by the Theorem \ref{corr} and the equation (6) in \cite{FV12}.

From Fig. \ref{fig_comm}, we exclude the BD\cite{BD10} and BGGJK-1\cite{BGGJK} methods to distribute many shares to a user. Our method is better than the BGGJK-2\cite{BGGJK} when the user number is greater than $26$. The EC-ElGamal method has the best communication performance when the number of users is greater than $20$.
\begin{figure}[!t]
\centering
\includegraphics[width=2.5in]{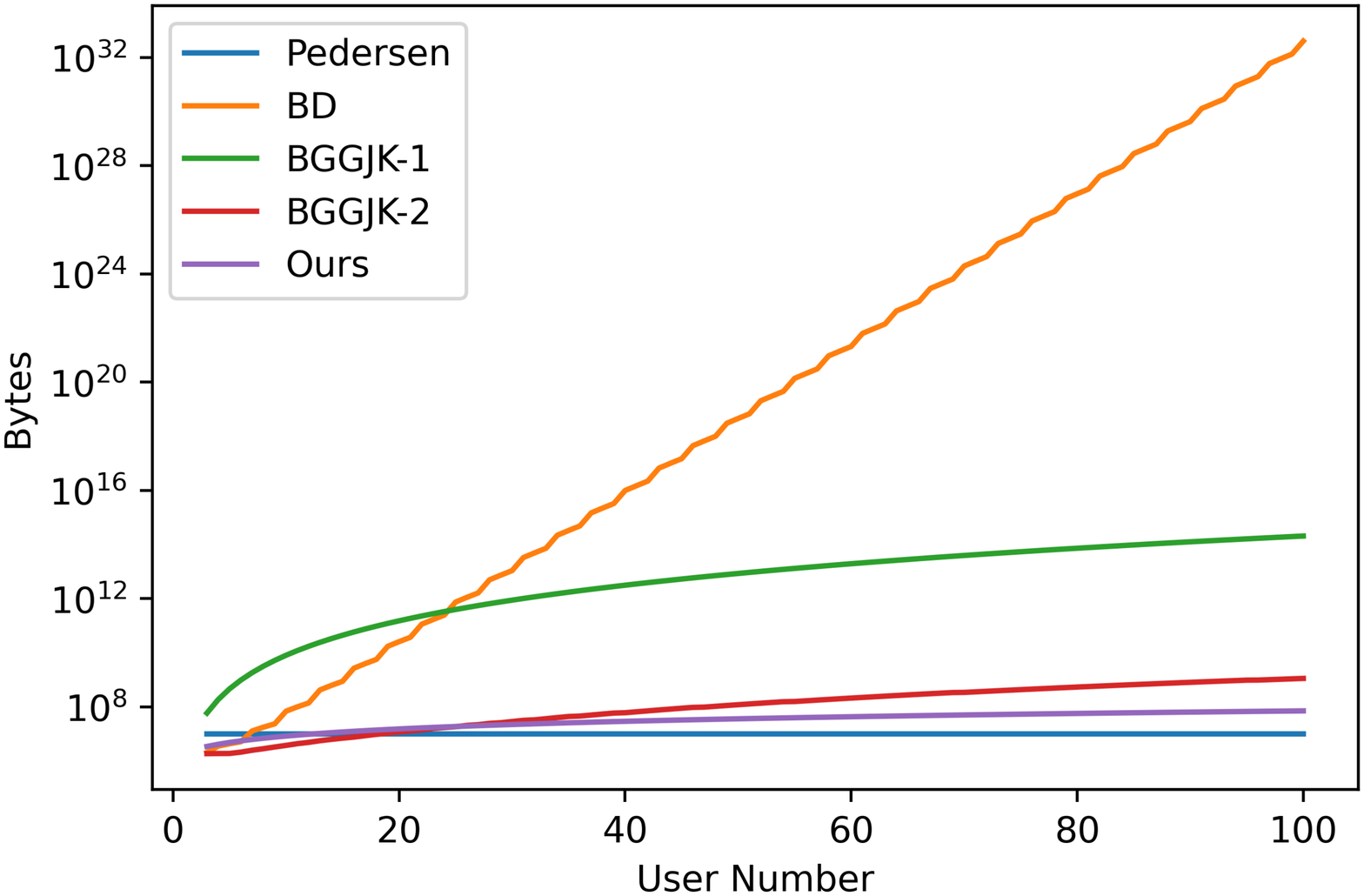}
\caption{Main Communication Overheads of Different DTAHE Schemes in the Protocol}
\label{fig_comm}
\end{figure}

\subsection{Computation}
Since the computation costs of the protocol are dominated by a DTAHE scheme. We give Table II to show the time cost of some DTAHE algorithms. We set the user number as $n = 35$ and implement three DTAHE schemes for comparison. The schemes are implemented in Python. The ``pyOpenSSL'' is used to implement the EC-ElGamal based DTAHE. The small discrete logarithm of a group element is found by the well-known ``Baby-Step-Giant-Step'' method. An open source library ``bfv-python'' is adapted to implement our scheme and the BGGJK-2 method in \cite{BGGJK}. For simplicity, we use a public python module ``eciespy'' as an instance of the $HPKE$ scheme. The CPUs are Intel(R) Core(TM) i7-8550U (1.80GHz, 1.99GHz) and the RAM is 16GB.

We use the ``secp256r1'' curve as the parameter set to implement the EC-ElGamal based DTAHE. For the security parameter $\lambda=128$, we set $d =2048$, $|h|=54$ and $|l|=17$ for our DTAHE scheme. When $n=35$, the scheme constructed by the BGGJK-2 method in \cite{BGGJK} requires $|h| \geq 426$. So we set $d = 16384$ according to the parameter table in \cite{fhestd}. Each data element in $m_u$ in our test occupies $8$ bits. The time costs of the three implementations are listed in the table II, which are measured in seconds. Apparently, a secure linear aggregation protocol with the EC-ElGamal based DTAHE takes the least time on user side. The protocol with our DTAHE scheme takes the least time on server side.

\begin{table}[!t]%%!t
\label{tabcomp}
\renewcommand{\arraystretch}{1.3}
\centering
\setlength{\abovecaptionskip}{0pt}%
\setlength{\belowcaptionskip}{10pt}%
\caption{Computation Costs of Some DTAHE Algorithms with Fixed Number of Users}
\resizebox{0.49\textwidth}{15mm}{
\begin{tabular}{|c|c|c|c|c|c|}
\hline
                  &  $Share$      & $CombKey$ and $Enc$             &  $Eval$                 &  $ParDec$   & $FinDec$  \\  \hline
\cite{P91}        & $0.03$           & $6.95$                 & $14.62$                        & $6.32$ & $441.25$ \\  \hline
\cite{BGGJK}-2    & $4.03$          & $15.48$                 & $2.40$                             &$17.29$  &  $2.98$\\  \hline
Ours              & $17.46$        & $7.50$                   & $1.52$                        &$51.17$    & $1.64$ \\ \hline   %%length of noise
\end{tabular}
}
\end{table}

\section{Conclusion}
This paper shows a secure linear protocol mainly for complex federated learning models. When the communication cost is not the dominate factor, the protocol could be deployed in a federated learning model to protect the private inputs of users. The DTAHE schemes in the protocol may be further optimized to reduce the computation time. For example, parallel computing technologies may reduce the time cost of the EC-ElGamal based DTAHE, and multi-secret sharing schemes may reduce the communication and computation costs of our DTAHE scheme.

% or


%\appendices

%\section{}

\section*{Acknowledgment}

This work is supported by the National Key R\&D Program of China (2017YFB0802500), Guangdong Major Project of  Basic and Applied Basic Research(2019B030302008) and Natural Science Foundation of Guangdong Province of China (2018A0303130133).

%\ifCLASSOPTIONcaptionsoff
%  \newpage
%\fi



%       \bibliographystyle{IEEEtran}
%\bibliographystyle{plain}
\bibliography{mybib}
\end{document}